\documentclass[a4paper]{article}
\usepackage{cite}
\usepackage[english]{babel}
\usepackage[utf8x]{inputenc}
\usepackage[T1]{fontenc}
\usepackage{pgf,tikz}
\usepackage{amsfonts}
\usetikzlibrary{arrows}
\usepackage[bottom]{footmisc}
\usepackage{appendix}
\usepackage[a4paper,top=3cm,bottom=3cm,left=3cm,right=3cm,marginparwidth=2cm]{geometry}

\usepackage{amsthm}
\newtheorem{thm}{Theorem}
\newtheorem{conjecture}{Conjecture}
\usepackage{amsmath}
\usepackage{bbm}
\usepackage{graphicx}
\usepackage[colorinlistoftodos]{todonotes}
\usepackage[colorlinks=true, allcolors=blue]{hyperref}
\usepackage{pgfplots} 
\usepackage{authblk}
\pgfplotsset{grid style={gridColor,line width=0pt}}
\newcommand\f{\frac}

\def\be#1{\begin{equation}\label{#1}}
\def\ee{\end{equation}}

\title{A mathematical model of p62-ubiquitin aggregates in autophagy}
\author[1]{Julia Delacour}
\author[1]{Marie Doumic}
\author[2]{Sascha Martens}
\author[3]{Christian Schmeiser}
\author[4]{Gabriele Zaffagnini} 
\affil[1]{Sorbonne Universit\'e, Inria, Universit\'e Paris-Diderot, CNRS, Laboratoire Jacques-Louis Lions, 75005 Paris, France}
\affil[2]{University of Vienna, Max F. Perutz Laboratories Vienna Biocenter (VBC), Dr. Bohr-Gasse 9, 1030 Vienna, Austria}
\affil[3]{Faculty of Mathematics, University of Vienna, Oskar-Morgenstern-Platz 1,1090 Vienna, Austria}
\affil[4]{Centre for Genomic Regulation (CRG), Dr.  Aiguader 88, 08003 Barcelona, Spain}
\begin{document}

\maketitle
\begin{abstract}
Aggregation of ubiquitinated cargo by oligomers of the protein p62 is an important preparatory step in cellular autophagy.
In this work a mathematical model for the dynamics of these heterogeneous aggregates in the form of a system of ordinary 
differential equations is derived and analyzed. Three different parameter regimes are identified, where either aggregates are
unstable, or their size saturates at a finite value, or their size grows indefinitely as long as free particles are abundant. 
The boundaries of these regimes as well as the finite size in the second case can be computed explicitly. The growth in the
third case (quadratic in time) can also be made explicit by formal asymptotic methods. The qualitative results are illustrated
by numerical simulations. A comparison with recent experimental results permits a partial parametrization of the 
model.
\end{abstract}

\paragraph{Acknowledgments:} This work has been supported by the PhD program {\em Signalling Mechanisms in Cellular
Autophagy,} funded by the Austrian Science Fund (FWF), project no. W1261. CS also acknowledges support by FWF,
grant nos. W1245 and SFB65. MD and JD have been partially supported by the ERC Starting Grant SKIPPERAD (number 306321). MD thanks the Wolfgang Pauli Institute for the sabbatical stay in Vienna during which this work has been initiated.

\section{Introduction}
Autophagy is an intracellular pathway, which targets damaged, surplus, and harmful cytoplasmic material for degradation. This is mediated by the sequestration of cytoplasmic cargo material within double membrane vesicles termed autophagosomes, which subsequently fuse with lysosomes wherein the cargo is hydrolyzed.  Defects in autophagy result in various diseases including neurodegeneration, cancer, and uncontrolled infections \cite{cite1}.
The selectivity of autophagic processes is mediated by cargo receptors such as p62 (also known as SQSTM1), which link the cargo material to the nascent autophagosomal membrane \cite{Danielijcs214304}.  p62 is an oligomeric protein and mediates the selective degradation of ubiquitinated proteins. Its interaction with ubiquitin is mediated by its C-terminal UBA domain, while it attaches the cargo to the autophagosomal membrane due to its interaction with Atg8 family proteins such as LC3B, which decorate the membrane \cite{Pankiv17082007}. Additionally, p62 serves to condensate ubiquitinated proteins into larger condensates or aggregates, which subsequently become targets for autophagy \cite{cite2,martens2}. It has been reported that this condensation reaction requires the ability of p62 to oligomerize and the presence of two or more ubiquitin chains on the substrates \cite{10.7554/eLife.08941,martens2}.

In this article a mathematical model for the condensation process is derived and analyzed. It is based on cross-linking of
p62 oligomers by ubiquitinated substrate \cite{martens2}. A cross-linker is assumed to be able to connect two oligomers,
where each oligomer has a number of binding sites corresponding to its size. As an approximation for the dynamics of large aggregates, a nonlinear system of ordinary differential equations is derived.

The oligomerization property of p62 has been shown to be necessary in the formation of aggregates~\cite{martens2}: too small oligomers of Ubiquitin do not form aggregates~\cite{10.7554/eLife.08941}.

The dynamics of protein aggregation has been studied by mathematical modelling for several decades, but most models consider the aggregation of only one type of protein, which gives rise to models belonging to the class of nucleation-coagulation-fragmentation equations, see e.g.~\cite{Bishop,10.1371/journal.pone.0043273,Radford} for examples in the biophysical literature, and~\cite{Thierry,LM4,Banasiak,DubovStew} for a sample of the mathematical literature. 
Contrary to these studies, the present work considers aggregates composed of two different types of particles with varying
mixing ratios, which drastically increases the complexity of the problem. 

In the following section the mathematical model is derived. It describes an aggregate by three numbers: the number of
p62 oligomers, the number of cross-linkers bound to one oligomer, and the number of cross-linkers bound to two oligomers.
The model considers an early stage of the aggregation process where the supply of free p62 oligomers and of free cross-linkers is not limiting.
Since no other information about the composition of the aggregate is used, assumptions on the binding and unbinding rates
are necessary. In the limit of large aggregates, whose details are presented in an appendix, the model takes the form of a
system of three ordinary differential equations. Section 3 starts with a result on the well posedness of the model, and it is 
mainly devoted to a study of the long-time behaviour by a combination of analytical and numerical methods. Depending on the 
parameter values, three different regimes are identified, where either aggregates are unstable and completely dissolved, or
their size tends to a limiting value, or they keep growing (as long as they do not run out of free oligomers and cross-linkers). 
In Section 4 we discuss the parametrization of the model and a comparison with data from \cite{martens2}. 

\section{Presentation of the model}\label{sec:model}

\paragraph{Discrete description of aggregates:}
We consider two types of basic particles:
\begin{enumerate}
\item {\em Oligomers} of the protein p62, where we assume for simplicity that all oligomers contain the same number $n\ge 3$ 
of molecules. These oligomers are denoted by p62$_n$ and are assumed to possess $n$ binding sites for ubiquitin each,
\item {\em Cross-linkers} in the form of ubiquitinated cargo, denoted by $Ubi$ and assumed to have two ubiquitin ends each. When one end of a $Ubi$ is bound to a p62$_n$, we call it {\em one-hand bound,} when both ends are bound we call it {\em both-hand bound.}
\end{enumerate}

An aggregate is represented by a triplet $(i,j,k)\in \mathbb{N}_0^3$, where $i$ denotes the number of one-hand bound $Ubi$,
$j$ denotes the number of both-hand bound $Ubi$, and $k$ denotes the number of p62$_n$. It is a rather drastic step to 
describe an aggregate only by these three numbers, since the same triplet might represent aggregates with various forms.
This will affect our modelling below. 

An aggregate will be assumed to contain at least two p62$_n$, i.e. $k\ge 2$, and enough both-hand bound $Ubi$ to be 
connected, i.e. $j\ge k-1$. Furthermore, an aggregate contains $nk$ binding sites for $Ubi$, implying 
$i+2j \le nk$. A triplet $(i,j,k)\in \mathbb{N}_0^3$ satisfying the three inequalities 
\begin{equation}\label{discr-cond}
    k\ge 2 \,,\qquad j\ge k-1 \,,\qquad i+2j \le nk \,,
\end{equation}
will be called {\em admissible.} 
An example of an admissible triplet describing a unique aggregate shape is $(0,k-1,k)$, representing a chain of p62$_n$.
Adding one both-hand bound $Ubi$ already creates a shape ambiguity: The triplet $(0,k,k)$ can be realized by a circular
aggregate or by an open chain, where one connection is doubled.

\paragraph{The reaction scheme:} Basically there are only two types of reactions: binding and unbinding of $Ubi$ to 
p62$_n$. However, depending on the situation these may have various effects on the aggregate, whence we distinguish
between three binding and three unbinding reactions.
\begin{enumerate}
\item {\bf Addition of a free $Ubi$}, requiring at least one free binding site, i.e. $nk-i-2j\ge 1$, (see Fig.~\ref{fig:1}):
\begin{align*}
Ubi+ (i,j,k) \xrightarrow{\kappa_1'}(i+1,j,k)
\end{align*}
The reaction rate (number of reactions per time) is modeled by mass action kinetics for a second-order reaction with 
reaction constant $\kappa_1'$ and with the number $[Ubi]$ of free $Ubi$. Since free $Ubi$ and free p62 oligomers 
will be assumed abundant, their numbers $[Ubi]$ and $[{\rm p62}_n]$ will be kept fixed and the abbreviation 
$\kappa_1 = \kappa_1' [Ubi]$ will be used. This leads to a first-order reaction rate
\be{def:r1}
  r_1 = \kappa_1 (nk-i-2j) \,. 
\ee
\item {\bf Addition of a free} p62$_n$, requiring at least one one-hand bound $Ubi$, i.e. $i\ge 1$:
\begin{align*}
{\rm p62}_n+ (i,j,k) \xrightarrow{\kappa_2'}(i-1,j+1,k+1)
\end{align*}
Analogously to above, we set $\kappa_2 = \kappa_2' [{\rm p62}_n]$ and
\be{def:r2}
  r_2 = \kappa_2 i \,.
\ee

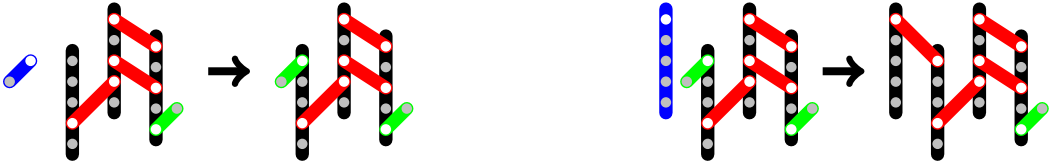
\begin{figure}[h!]
\begin{tikzpicture}[scale=0.275]
\draw[black, line width=5pt, cap=round]  (0,3)-- (0,-2);
\draw[line width=5pt, cap=round] (2,5)-- (2,0);
\draw[line width=5pt, cap=round] (4,3.7)-- (4,-1.3);

\draw[red, line width=5pt, cap=round] (4,3.2)-- (2,4.5);
\draw[red, line width=5pt, cap=round] (4,1.2)-- (2,2.5);
\draw[red, line width=5pt, cap=round] (0,-0.5)-- (2,1.5);
\draw[green, line width=5pt, cap=round] (4,-0.8)-- (5,0.2);
\draw[blue, line width=5pt, cap=round] (-3,1.5)-- (-2,2.5);

\fill[gray!50] (-3,1.5) circle (0.25);
\fill[white] (-2,2.5) circle (0.25);
\fill[gray!50] (5,0.2) circle (0.25);
\draw[->,line width=3pt] (6.5,2)--(8.5,2);
\fill[gray!50] (0,2.5) circle (0.25);
\fill[gray!50] (0,1.5) circle (0.25);
\fill[gray!50] (0,0.5) circle (0.25);
\fill[white] (0,-0.5) circle (0.25);
\fill[gray!50] (0,-1.5) circle (0.25);

\fill[white] (2,4.5) circle (0.25);
\fill[gray!50] (2,3.5) circle (0.25);
\fill[white] (2,2.5) circle (0.25);
\fill[white] (2,1.5) circle (0.25);
\fill[gray!50] (2,0.5) circle (0.25);

\fill[white] (4,3.2) circle (0.25);
\fill[gray!50] (4,2.2) circle (0.25);
\fill[white] (4,1.2) circle (0.25);
\fill[gray!50] (4,0.2) circle (0.25);
\fill[white] (4,-0.8) circle (0.25);

\draw[black, line width=5pt, cap=round]  (11,3)-- (11,-2);
\draw[black, line width=5pt, cap=round] (13,5)-- (13,0);
\draw[black, line width=5pt, cap=round] (15,3.7)-- (15,-1.3);
\draw[red, line width=5pt, cap=round] (11,-0.5)-- (13,1.5);
\draw[red, line width=5pt, cap=round] (15,3.2)-- (13,4.5);
\draw[red, line width=5pt, cap=round] (15,1.2)-- (13,2.5);
\draw[green, line width=5pt, cap=round] (15,-0.8)-- (16,0.2);
\draw[green, line width=5pt, cap=round] (11,2.5)-- (10,1.5);
\fill[white] (11,2.5) circle (0.25);
\fill[gray!50] (11,1.5) circle (0.25);
\fill[gray!50] (11,0.5) circle (0.25);
\fill[white] (11,-0.5) circle (0.25);
\fill[gray!50] (11,-1.5) circle (0.25);

\fill[white] (13,4.5) circle (0.25);
\fill[gray!50] (13,3.5) circle (0.25);
\fill[white] (13,2.5) circle (0.25);
\fill[white] (13,1.5) circle (0.25);
\fill[gray!50] (13,0.5) circle (0.25);

\fill[white] (15, 3.2) circle (0.25);
\fill[gray!50] (15,2.2) circle (0.25);
\fill[white] (15,1.2) circle (0.25);
\fill[gray!50] (15,0.2) circle (0.25);
\fill[white] (15,-0.8) circle (0.25);

\fill[gray!50] (16,0.2) circle (0.25);
\fill[gray!50] (10,1.5) circle (0.25);
\end{tikzpicture}
\hspace{3 cm}
\begin{tikzpicture}[scale=0.275]
\draw[blue, line width=5pt, cap=round]  (0,5)-- (0,0);
\draw[black, line width=5pt, cap=round]  (2,3)-- (2,-2);
\draw[line width=5pt, cap=round] (4,5)-- (4,0);
\draw[line width=5pt, cap=round] (6,3.7)-- (6,-1.3);

\draw[red, line width=5pt, cap=round] (6,3.2)-- (4,4.5);
\draw[red, line width=5pt, cap=round] (6,1.2)-- (4,2.5);
\draw[red, line width=5pt, cap=round] (2,-0.5)-- (4,1.5);
\draw[green, line width=5pt, cap=round] (6,-0.8)-- (7,0.2);
\draw[green, line width=5pt, cap=round] (2,2.5)-- (1,1.5);

\fill[white] (0,4.5) circle (0.25);
\fill[gray!50] (0,3.5) circle (0.25);
\fill[gray!50] (0,2.5) circle (0.25);
\fill[gray!50] (0,1.5) circle (0.25);
\fill[gray!50] (0,0.5) circle (0.25);

\fill[white] (2,2.5) circle (0.25);
\fill[gray!50] (1,1.5) circle (0.25);
\fill[gray!50] (2,1.5) circle (0.25);
\fill[gray!50] (2,0.5) circle (0.25);
\fill[white] (2,-0.5) circle (0.25);
\fill[gray!50] (2,-1.5) circle (0.25);

\fill[white] (4,4.5) circle (0.25);
\fill[gray!50] (4,3.5) circle (0.25);
\fill[white] (4,2.5) circle (0.25);
\fill[white] (4,1.5) circle (0.25);
\fill[gray!50] (4,0.5) circle (0.25);

\fill[white] (6,3.2) circle (0.25);
\fill[gray!50] (6,2.2) circle (0.25);
\fill[white] (6,1.2) circle (0.25);
\fill[gray!50] (6,0.2) circle (0.25);
\fill[white] (6,-0.8) circle (0.25);
\fill[gray!50] (7,0.2) circle (0.25);

\draw[->,line width=3pt] (7.5,2)--(9.5,2);
\draw[black, line width=5pt, cap=round]  (11,5)-- (11,0);
\draw[black, line width=5pt, cap=round]  (13,3)-- (13,-2);
\draw[line width=5pt, cap=round] (15,5)-- (15,0);
\draw[line width=5pt, cap=round] (17,3.7)-- (17,-1.3);

\draw[red, line width=5pt, cap=round] (17,3.2)-- (15,4.5);
\draw[red, line width=5pt, cap=round] (17,1.2)-- (15,2.5);
\draw[red, line width=5pt, cap=round] (13,-0.5)-- (15,1.5);
\draw[red, line width=5pt, cap=round] (13,2.5)-- (11,4.5);
\draw[green, line width=5pt, cap=round] (17,-0.8)-- (18,0.2);

\fill[white] (11,4.5) circle (0.25);
\fill[gray!50] (11,3.5) circle (0.25);
\fill[gray!50] (11,2.5) circle (0.25);
\fill[gray!50] (11,1.5) circle (0.25);
\fill[gray!50] (11,0.5) circle (0.25);

\fill[white] (13,2.5) circle (0.25);
\fill[gray!50] (13,1.5) circle (0.25);
\fill[gray!50] (13,0.5) circle (0.25);
\fill[white] (13,-0.5) circle (0.25);
\fill[gray!50] (13,-1.5) circle (0.25);

\fill[white] (15, 4.5) circle (0.25);
\fill[gray!50] (15,3.5) circle (0.25);
\fill[white] (15,2.5) circle (0.25);
\fill[white] (15,1.5) circle (0.25);
\fill[gray!50] (15,0.5) circle (0.25);

\fill[white] (17,3.2) circle (0.25);
\fill[gray!50] (17,2.2) circle (0.25);
\fill[white] (17,1.2) circle (0.25);
\fill[gray!50] (17,0.2) circle (0.25);
\fill[white] (17,-0.8) circle (0.25);
\fill[gray!50] (18,0.2) circle (0.25);
\end{tikzpicture}
\label{fig:1}
\caption{Examples for Reactions 1 (left) and 2 (right) with p62$_5$ in black, one-hand bound $Ubi$ in green, two-hand
bound $Ubi$ in red, free particles in blue. Reaction 1: $Ubi + (1,3,3) \to (2,3,3)$. Reaction 2: p62$_5 + (2,3,3) \to (1,4,4)$.}
\end{figure}

\item {\bf Compactification of the aggregate} by a $Ubi$ binding its second hand, requiring at least one one-hand bound $Ubi$,
i.e. $i\ge 1$, and at least one free binding site, i.e. $nk-i-2j\ge 1$:
\begin{align*}
(i,j,k) \xrightarrow{\kappa_3'} (i-1,j+1,k)
\end{align*}
This is a second-order reaction with rate
\be{def:r3}
  r_3 = \kappa_3' i(nk-i-2j) \,.
\ee
\item {\bf Loss of a} $Ubi$, requiring at least one one-handed $Ubi$, i.e. $i\ge 1$. This is the reverse reaction to 1:
\begin{align*}
(i,j,k) &\xrightarrow{\kappa_{-1}} Ubi + (i-1,j,k) 
\end{align*}
Its rate is modeled by
\be{def:r-1}
  r_{-1} = \kappa_{-1} i \,.
\ee
\item {\bf Loss of a} p62$_n$ (leading to loss of the whole aggregate if $k=2$):
\begin{align*}
(i,j,k) &\xrightarrow{\kappa_- \alpha_{j,k}} {\rm p62}_n + \ell\,Ubi + (i+1-\ell,j-1,k-1) 
\end{align*}
This and the following reaction need some comments. They are actually both the same reaction, namely breaking
of a cross-link, which we assume to occur with rate $\kappa_- j$. However, this can have different consequences. 
Here we consider something close to the reverse of reaction 2. This means we assume that the broken cross-link has 
been the only connection of a p62 oligomer with the aggregate, such that the oligomer falls off. This requires the
condition $nk-2j\ge n-1$, meaning the possibility that the other $n-1$ binding sites of the lost oligomer are free of
two-hand bound $Ubi$. It is not quite the reverse 
of reaction 2, since we have to consider the possibility that $\ell$ one-hand bound $Ubi$, $0\le \ell \le n-1$, are bound to 
the lost oligomer. The conditional probability $\alpha_{j,k}$ to be in this case, 
when a cross-link breaks, is {\em zero} for a very tightly connected aggregate where each oligomer is cross-linked at least
twice, i.e. $nk-2j \le n-2$, and it is {\em one} for a very
loose aggregate, i.e. a chain with $j=k-1$. This leads to the model
\be{def:alpha-discr}
   \alpha_{j,k} = \frac{(nk-2j-n+2)_+}{(n-2)k+4-n} \,,
\ee
and to the rate
\be{def:r-2}
  r_{-2} = \kappa_- \alpha_{j,k} j \,.
\ee
In the framework of our model, $\ell$ should be a random number satisfying the restrictions
\begin{equation}\label{ell-cond}
  (n-1-nk+i+2j)_+ \le \ell \le \min\{i,n-1\} \,,
\end{equation}
where the upper bound should be obvious and the lower bound implies that the last condition in \eqref{discr-cond} is
satisfied after the reaction. We shall use the choice
\be{def:ell}
   \ell = \ell_{i,j,k} := \left\lfloor\frac{(n-1)i}{nk-2j} \right\rceil\,,
\ee
which can be interpreted as the rounded ($\lfloor\cdot\rceil$ denotes the closest integer) expectation value for the number 
of one-hand bound $Ubi$ on the lost oligomer in terms
of the ratio between the number $n-1$ of available binding sites on the lost oligomer and the total number $nk-2j$ of
available binding sites for one-hand bound $Ubi$ in the whole aggregate. It is easily seen that in the relevant situation 
$\alpha_{j,k}>0$, i.e. $nk-2j\ge n-1$, the choice \eqref{def:ell} without the rounding satisfies the conditions \eqref{ell-cond}.
Since the bounds in \eqref{ell-cond} are integer, the same is true for the rounded version.

Note that we neglect the possibility to lose more than one oligomer by breaking a cross-link,
i.e. the fragmentation of the aggregate into two smaller ones. This is a serious and actually questionable modelling assumption.
An a posteriori justification will be provided by some of the results of the following section, showing that growing aggregates
are tightly connected.
\item {\bf Loosening of the aggregate} by breaking a cross-link, requiring at least one excess both-hand bound $Ubi$,
i.e. $j\ge k$:
\begin{align*}
(i,j,k) &\xrightarrow{\kappa_- (1-\alpha_{j,k})} (i+1,j-1,k) \,.
\end{align*}
This is the reverse of reaction 3 with the rate
\be{def:r-3}
  r_{-3} = \kappa_- (1-\alpha_{j,k})j \,,
\ee
which respects the requirement $j\ge k$ for a positive rate, because of
$$
  1 - \alpha_{j,k} = \min\left\{ 1, \frac{2(j-k+1)}{(n-2)k+4-n} \right\}\,.
$$
\end{enumerate}

\begin{figure}[h!]
\begin{tikzpicture}[scale=0.275]
\draw[black, line width=5pt, cap=round]  (2,3)-- (2,-2);
\draw[line width=5pt, cap=round] (4,5)-- (4,0);
\draw[line width=5pt, cap=round] (6,3.7)-- (6,-1.3);

\draw[red, line width=5pt, cap=round] (6,3.2)-- (4,4.5);
\draw[red, line width=5pt, cap=round] (6,1.2)-- (4,2.5);
\draw[red, line width=5pt, cap=round] (2,-0.5)-- (4,1.5);
\draw[green, line width=5pt, cap=round] (6,-0.8)-- (7,0.2);
\draw[green, line width=5pt, cap=round] (2,2.5)-- (1,1.5);

\fill[white] (2,2.5) circle (0.25);
\fill[gray!50] (2,1.5) circle (0.25);
\fill[gray!50] (2,0.5) circle (0.25);
\fill[white] (2,-0.5) circle (0.25);
\fill[gray!50] (2,-1.5) circle (0.25);

\fill[white] (4,4.5) circle (0.25);
\fill[gray!50] (4,3.5) circle (0.25);
\fill[white] (4,2.5) circle (0.25);
\fill[white] (4,1.5) circle (0.25);
\fill[gray!50] (4,0.5) circle (0.25);

\fill[white] (6,3.2) circle (0.25);
\fill[gray!50] (6,2.2) circle (0.25);
\fill[white] (6,1.2) circle (0.25);
\fill[gray!50] (6,0.2) circle (0.25);
\fill[white] (6,-0.8) circle (0.25);
\fill[gray!50] (7,0.2) circle (0.25);
\fill[gray!50] (1,1.5) circle (0.25);

\draw[->,line width=3pt] (8.5,2)--(10.5,2);
\draw[black, line width=5pt, cap=round]  (13,3)-- (13,-2);
\draw[line width=5pt, cap=round] (15,5)-- (15,0);
\draw[line width=5pt, cap=round] (17,3.7)-- (17,-1.3);

\draw[red, line width=5pt, cap=round] (17,3.2)-- (15,4.5);
\draw[red, line width=5pt, cap=round] (17,1.2)-- (15,2.5);
\draw[red, line width=5pt, cap=round] (13,-0.5)-- (15,1.5);
\draw[red, line width=5pt, cap=round] (17,-0.8)-- (13,-1.5);
\draw[green, line width=5pt, cap=round] (13,2.5)-- (12,1.5);

\fill[gray!50] (12,1.5) circle (0.25);
\fill[white] (13,2.5) circle (0.25);
\fill[gray!50] (13,1.5) circle (0.25);
\fill[gray!50] (13,0.5) circle (0.25);
\fill[white] (13,-0.5) circle (0.25);
\fill[white] (13,-1.5) circle (0.25);

\fill[white] (15, 4.5) circle (0.25);
\fill[gray!50] (15,3.5) circle (0.25);
\fill[white] (15,2.5) circle (0.25);
\fill[white] (15,1.5) circle (0.25);
\fill[gray!50] (15,0.5) circle (0.25);

\fill[white] (17,3.2) circle (0.25);
\fill[gray!50] (17,2.2) circle (0.25);
\fill[white] (17,1.2) circle (0.25);
\fill[gray!50] (17,0.2) circle (0.25);
\fill[white] (17,-0.8) circle (0.25);

\end{tikzpicture}
\hspace{3cm}
\begin{tikzpicture}[scale=0.275]

\draw[black, line width=5pt, cap=round]  (0,3)-- (0,-2);
\draw[line width=5pt, cap=round] (2,5)-- (2,0);
\draw[line width=5pt, cap=round] (4,3.7)-- (4,-1.3);

\draw[red, line width=5pt, cap=round] (4,3.2)-- (2,4.5);
\draw[red, line width=5pt, cap=round] (4,1.2)-- (2,2.5);
\draw[red, line width=5pt, cap=round] (0,-0.5)-- (2,1.5);
\draw[green, line width=5pt, cap=round] (4,-0.8)-- (5,0.2);

\fill[gray!50] (0,2.5) circle (0.25);
\fill[gray!50] (0,1.5) circle (0.25);
\fill[gray!50] (0,0.5) circle (0.25);
\fill[white] (0,-0.5) circle (0.25);
\fill[gray!50] (0,-1.5) circle (0.25);

\fill[white] (2,4.5) circle (0.25);
\fill[gray!50] (2,3.5) circle (0.25);
\fill[white] (2,2.5) circle (0.25);
\fill[white] (2,1.5) circle (0.25);
\fill[gray!50] (2,0.5) circle (0.25);

\fill[white] (4,3.2) circle (0.25);
\fill[gray!50] (4,2.2) circle (0.25);
\fill[white] (4,1.2) circle (0.25);
\fill[gray!50] (4,0.2) circle (0.25);
\fill[white] (4,-0.8) circle (0.25);
\fill[gray!50] (5,0.2) circle (0.25);
\draw[->,line width=3pt] (7,2.5)--(9,4);
\draw[->,line width=3pt] (7,1.5)--(9,0);
\draw[line width=5pt, cap=round] (13,8)-- (13,3);
\draw[line width=5pt, cap=round] (15,6.7)-- (15,1.7);

\draw[red, line width=5pt, cap=round] (15,6.2)-- (13,7.5);
\draw[red, line width=5pt, cap=round] (15,4.2)-- (13,5.5);
\draw[green, line width=5pt, cap=round] (15,2.2)-- (16,3.2);
\draw[green, line width=5pt, cap=round] (13,4.5)-- (12,3.5);
\fill[white] (13,7.5) circle (0.25);
\fill[gray!50] (13,6.5) circle (0.25);
\fill[white] (13,5.5) circle (0.25);
\fill[white!50] (13,4.5) circle (0.25);
\fill[gray!50] (13,3.5) circle (0.25);
\fill[gray!50] (12,3.5) circle (0.25);

\fill[white] (15,6.2) circle (0.25);
\fill[gray!50] (15,5.2) circle (0.25);
\fill[white] (15,4.2) circle (0.25);
\fill[gray!50] (15,3.2) circle (0.25);
\fill[white] (15,2.2) circle (0.25);
\fill[gray!50] (16,3.2) circle (0.25);

\draw[line width=5pt, cap=round] (11,0)-- (11,-5);
\draw[line width=5pt, cap=round] (13,2)-- (13,-3);
\draw[line width=5pt, cap=round] (15,0.7)-- (15,-4.3);

\draw[red, line width=5pt, cap=round] (11,-3.5)-- (13,-1.5);
\draw[red, line width=5pt, cap=round] (15,-1.8)-- (13,-0.5);
\draw[green, line width=5pt, cap=round] (15,-3.8)-- (16,-2.8);
\draw[green, line width=5pt, cap=round] (13,1.5)-- (14,0.8);
\fill[gray!50] (11,-0.5) circle (0.25);
\fill[gray!50] (11,-1.5) circle (0.25);
\fill[gray!50] (11,-2.5) circle (0.25);
\fill[white] (11,-3.5) circle (0.25);
\fill[gray!50] (11,-4.5) circle (0.25);

\fill[white!50] (13,1.5) circle (0.25);
\fill[gray!50] (13,0.5) circle (0.25);
\fill[white] (13,-0.5) circle (0.25);
\fill[white] (13,-1.5) circle (0.25);
\fill[gray!50] (13,-2.5) circle (0.25);
\fill[gray!50] (14,0.8) circle (0.25);

\fill[gray!50] (15,0.2) circle (0.25);
\fill[gray!50] (15,-0.8) circle (0.25);
\fill[white] (15,-1.8) circle (0.25);
\fill[gray!50] (15,-2.8) circle (0.25);
\fill[white] (15,-3.8) circle (0.25);
\fill[gray!50] (16,-2.8) circle (0.25);
\end{tikzpicture}
\caption{Examples for Reaction 3 (left, $(2,3,3)\to(1,4,3)$), Reaction 5 (right, up, $(1,3,3)\to {\rm p62}_5 + (2,2,2)$, $\ell=0$), and Reaction 6 (right, down, $(1,3,3)\to(2,2,3)$).  \label{fig:reac3}}
\end{figure}
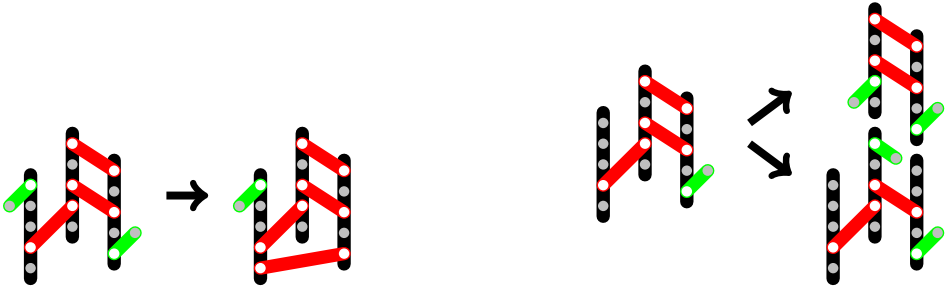

\paragraph{A deterministic model for large aggregates:}
The next step is the formulation of an evolution problem for a probability density on the set of admissible states $(i,j,k)$.
In this problem the discrete state is scaled by a typical value $k_0$ of $[Ubi]$ and $[{\rm p62}_n]$, assumed of the same
order of magnitude: 
\be{def:pqr}
    p := \frac{i}{k_0} \,,\qquad q := \frac{j}{k_0} \,,\qquad r := \frac{k}{k_0} \,.
\ee
It is then consistent with the definitions of $\kappa_1$ and $\kappa_2$ above to introduce $\kappa_3 := \kappa_3' k_0$.
In the large aggregate limit $k_0\to\infty$, the new unknowns become continuous, and the equation for the probability
density becomes a transport equation (see Appendix~\ref{app:model} for the details). It possesses deterministic solutions
governed by the ODE initial value problem
\begin{equation} \label{eq_1}
\begin{aligned}
  \dot p &= (\kappa_1 - \kappa_3 p) (nr-p-2q) + \kappa_{-} q\left(1- \f{(n-1)p}{(n-2)r}\right) - (\kappa_2 + \kappa_{-1})p \,, &\qquad p(0)=p_0 \,,\\
  \dot q &= \kappa_2 p + \kappa_3 p(nr-p-2q) - \kappa_{-}q \,, &\qquad q(0)=q_0 \,,  \\
  \dot r &= \kappa_2 p - \kappa_{-} q \alpha(q,r) \,, &\qquad r(0)=r_0 \,,
  \end{aligned}
 \end{equation}
where 
\begin{equation}\label{def:alpha}
    \alpha(q,r):=\f{nr-2q}{(n-2)r} 
\end{equation}
is the limit of $\alpha_{j,k}$ as $k_0\to\infty$.
The conditions for admissible states $(p,q,r)\in [0,\infty)^2\times (0,\infty)$ are obtained in the limit of \eqref{discr-cond}: 
\be{cont-cond}
    s:= nr-p-2q \geq 0 \,,\qquad q \geq r \,,
\ee
implying, as expected,
\begin{align}
   0 \leq \alpha(q,r) \leq 1 \,. \label{alpha}
\end{align}
The equations satisfied by $s$ and $q-r$,
\begin{eqnarray}
\label{eq:s}
\dot s &=& (n-1)\kappa_2 p +\kappa_{-1} p +\kappa_{-} q \f{2(q-r)}{(n-2)r}
   - s\left(\kappa_3 p + \kappa_1 + \kappa_{-} q \f{n-1}{(n-2)r}\right) \,,\\
(q-r)\dot\,  &=& \kappa_3 p s - \f{2\kappa_- q}{(n-2)r}(q-r) \,,\label{eq:q-r}
\end{eqnarray}
show that the conditions \eqref{cont-cond} are propagated by \eqref{eq_1}.

\section{Analytic results \label{section 3}}

\paragraph{Global existence:}
Since the right hand sides of \eqref{eq_1} contain quadratic nonlinearities, it seems possible that solutions blow up in finite
time. On the other hand, the right hand sides are not well defined for $r=0$. The essence of the following global existence
result is that neither of these difficulties occurs.

\begin{thm}\label{thm 1} 
Let $3\le n\in\mathbb{N}$ and $\kappa_1,\kappa_2,\kappa_3,\kappa_{-1},\kappa_{-} \geq 0$. Let $(p_0,q_0,r_0) \in (0,\infty)^3$
satisfy \eqref{cont-cond}. Then problem~\eqref{eq_1} has a unique global solution satisfying $(p(t),q(t),r(t))\in (0,\infty)^3$
as well as \eqref{cont-cond} for any $t>0$. Also the following estimates hold for $t>0$: 
\begin{align} 
&p(t)+ q(t)+r(t) \leq (p_0+q_0+r_0) \exp\left(t\,\max\{\kappa_1 n,\kappa_2\}\right)  \,, \label{thm1:1}\\
&r(t) \geq \frac{2}{n}q(t) \ge \frac{2q_0}{n} \exp (-\kappa_{-} t)  \,. \label{thm1:2}
 \end{align}
\end{thm}

\begin{proof}
Local existence and uniqueness is a consequence of the Picard-Lindel\"of theorem. Global existence will follow from the 
bounds stated in the theorem. Positivity of the solution components, of $s=nr-p-2q$, and of $q-r$ is an immediate
consequence of the form of the equations \eqref{eq_1}, \eqref{eq:s}, \eqref{eq:q-r}. This also implies 
$$
    \dot p + \dot q + \dot r \le \kappa_1 nr + \kappa_2 p \le \max\{\kappa_1 n, \kappa_2\}(p+q+r) \,,
$$
which shows \eqref{thm1:1} by the Gronwall lemma. With \eqref{cont-cond}, the equation for $q$ in \eqref{eq_1} implies
$$
   \dot q \ge - \kappa_- q \,,
$$
and another application of the Gronwall lemma and of \eqref{cont-cond} proves \eqref{thm1:2} and, thus, completes the 
proof of the theorem.
 \end{proof}
 
 \paragraph{Long-time behaviour:}
 The first step in the long-time analysis is the investigation of steady states. Although the right hand sides of \eqref{eq_1}
 are not well defined for $r=0$, the origin $p=q=r=0$ can be considered as a steady state since
 $$
    0\le\alpha(q,r)\le 1 \qquad\mbox{and}\qquad \frac{p}{r} \le n
 $$
 hold for admissible states satisfying \eqref{cont-cond}. The origin is the only acceptable steady state with $r=0$, since
 $\alpha(q,r)$ and $p/r$ are not well defined in this case, so the factor $q$, multiplying them in the equations, needs to be
 zero. Finally, for a steady state this implies also $p=0$. The following result shows that at most one other steady state
 is possible which, somewhat miraculously, can be computed explicitly.
 
 \begin{thm}\label{thm:steady-state}
 Let $3\le n\in\mathbb{N}$, $\kappa_1,\kappa_2,\kappa_3,\kappa_{-1},\kappa_- >0$, and let
 \be{def:baralpha}
   \bar \alpha := \frac{n}{n-2} + \frac{\kappa_{-1} + \kappa_1 - \sqrt{(\kappa_1 + \kappa_{-1})^2 + 4  \kappa_1  
       \kappa_2 (n-1)}}{\kappa_{-} (n-1)}  
\ee
satisfy $0<\bar\alpha <1$. Then there exists an admissible steady state $(\bar p,\bar q,\bar r)\in (0,\infty)^3$ of \eqref{eq_1}
given by
\begin{eqnarray*}
  \bar p &=& \frac{\kappa_1 \kappa_2 (n-2)}{\kappa_3 (\kappa_- \hat q (n-1) 
     + \kappa_{-1} (n-2))}\, \frac{1-\bar\alpha}{\bar\alpha} \,,\\
  \bar q &=& \frac{\kappa_1 \kappa_2^2 (n-2)}{\kappa_3 \kappa_-(\kappa_- \hat q (n-1) 
     + \kappa_{-1} (n-2))}\, \frac{1-\bar\alpha}{\bar\alpha^2} \,,\\
  \bar r &=& \frac{\kappa_1 \kappa_2^2 (n-2)}{\hat q\kappa_3 \kappa_-(\kappa_- \hat q (n-1) 
     + \kappa_{-1} (n-2))}\, \frac{1-\bar\alpha}{\bar\alpha^2} \,,\\
\end{eqnarray*}
with $\bar\alpha = \alpha(\bar q,\bar r)$ and $\hat q = (n - (n-2)\bar\alpha)/2 \in (1,n/2)$. There exists no other steady state (besides the origin).
 \end{thm}
 
\begin{proof}
Assuming $\bar r> 0,$ we introduce
\begin{equation}
\hat p = \frac{\bar p} {\bar r} \label{p_hat} \,, \qquad\qquad \hat q =  \frac{\bar q} {\bar r} \,,
\end{equation}
and rewrite the steady state equations in terms of $\hat p$ and $\hat q$:
\begin{align}
0 &= (\kappa_1-\kappa_3 \bar p)(n-\hat p-2\hat q) + \kappa_{-} \hat q \left(1 - \hat p\frac{n-1}{n-2}\right) -(\kappa_2 + \kappa_{-1})\hat p \label {eq_1_bis} \,, \\
0 &= \kappa_2 \hat p + \kappa_3 \bar p (n-\hat p-2\hat q) -\kappa_{-} \hat q  \label{eq_2_bis} \,, \\
0 &= \kappa_2 \hat p -\kappa_{-} \hat q \bar \alpha \,,\qquad \mbox{with } \bar \alpha=\frac{n-2\hat q}{n-2} \,. \label{eq_3-bis} 
\end{align} 
From \eqref{eq_3-bis} we obtain
\begin{align}
\hat p = \frac{\kappa_- \hat q}{\kappa_2} \bar\alpha = \frac{\kappa_- \hat q(n-2\hat q)}{\kappa_2 (n-2)} \,, \label{eq_5}
\end{align}
which is substituted into the sum of \eqref{eq_1_bis} and \eqref{eq_2_bis}:
\begin{align*}
  (n - 2\hat q)\left(\kappa_1 - \frac{\kappa_1\kappa_-}{\kappa_2(n-2)}\hat q 
  - \frac{\kappa_{-}^2 (n-1)}{\kappa_2 (n-2)^2}\hat q^2 - \frac{\kappa_{-1}  \kappa_-}{\kappa_2(n-2)}\hat q\right) = 0 \,.
\end{align*}
The option $n=2\hat q$ leads to $\bar\alpha=0$, implying $\hat p = 0$ and, thus, $\bar p=0$, which contradicts 
\eqref{eq_2_bis}. Therefore the second paranthesis has to vanish, leading to a quadratic equation for $\hat q$ with
the only positive solution
$$
    \hat q = \frac{(n-2)\left(-\kappa_{-1} - \kappa_1 + \sqrt{(\kappa_1 + \kappa_{-1})^2  + 4  \kappa_1 \kappa_2 (n-1)}\right)}
    {2 \kappa_{-} (n-1)} \,. 
$$
Now \eqref{eq_3-bis} implies the formula for $\bar\alpha$ stated in the theorem and we note that $0<\bar \alpha <1$ implies 
$1<\hat q<n/2$. We compute $\hat p$ from $\hat q$ by~\eqref{eq_5} and note that $\hat p>0$ since $\bar\alpha>0$. 
We then compute $\hat s = \bar s/\bar r=n-\hat p-2\hat q$ from the sum of \eqref{eq_1_bis} and \eqref{eq_2_bis}:
$$
    \hat s = \hat p \f{\kappa_{-1}(n-2) +\kappa_- \hat q (n-1)}{(n-2)\kappa_1} 
      = \f{\kappa_- \hat q \left(\kappa_{-1}(n-2) +\kappa_- \hat q (n-1)\right)}{(n-2)\kappa_1 \kappa_2} \bar\alpha\,,
$$
which proves $\hat s>0$. Finally we obtain the formula for $\bar p$ from \eqref{eq_2_bis} as well as 
$\bar r=\bar p/\hat p$ and $\bar q=\bar r \hat q$. 
\end{proof}

For convenience below, the conditions in the theorem are made more explicit in terms of the parameters by
\begin{align}
\bar \alpha <1 \quad &\Leftrightarrow & \hat q >1 \quad & \Leftrightarrow & \kappa_1\kappa_2
  > \f{\kappa_-}{n-2}\left( \kappa_1 + \f{n-1}{n-2}\kappa_- +\kappa_{-1}\right ) \,,\label{def:f1} \\
\bar \alpha>0 \quad & \Leftrightarrow & \hat q <\f{n}{2} \quad& \Leftrightarrow & \kappa_1\kappa_2
   < \f{\kappa_- n}{2(n-2)}\left(\kappa_1 + \f{n(n-1)}{2(n-2)} \kappa_-+\kappa_{-1}\right) \,. \label{def:f2} 
\end{align}
The steady state approaches the origin $p=q=r=0$ as $\bar\alpha \to 1$, whereas all its components become unbounded
as $\bar\alpha\to 0$. This motivates the following.

\begin{conjecture}\label{conj:cases}
With the notation of Theorem \ref{thm:steady-state},
\begin{enumerate}
\item if $0<\bar\alpha<1$, then all solutions of \eqref{eq_1} converge to $(\bar p,\bar q,\bar r)$ as $t\to\infty$,
\item if $\bar\alpha\ge1$, then all solutions of \eqref{eq_1} converge to $(0,0,0)$ as $t\to\infty$,
\item if $\bar\alpha\le 0$, then for all solutions of \eqref{eq_1} we have $p(t),q(t),r(t) \to \infty$ as $t\to\infty$.
\end{enumerate}
\end{conjecture}

The conjecture has been supported by numerical simulations. Figures \ref{steady state}, \ref{zero ss}, and 
\ref{growing polynomial} show typical simulation results corresponding to the three cases. The conjecture is open,
and its proof is not expected to be easy. Note for example that not even the local stability of the origin in Case 2 can be 
investigated by standard methods, since the right hand side of \eqref{eq_1} lacks sufficient smoothness. Partial results will 
be published in separate work.

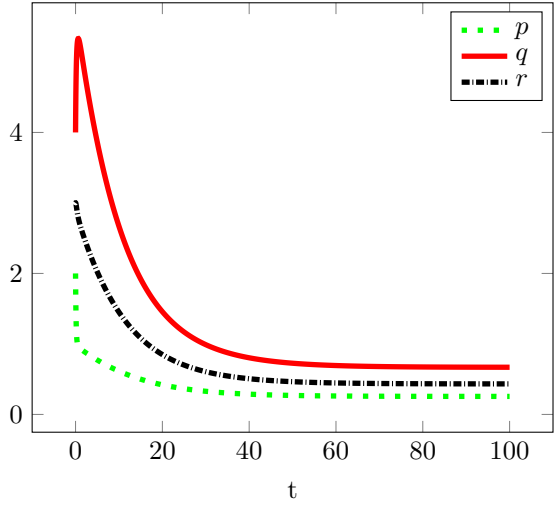
\begin{figure}[h!]
\centering
\begin{tikzpicture}
		\begin{axis}[xlabel={t}],
			\addplot [color=green,loosely dotted,line width=2.0pt] table[x={t}, y={x}] {donnees_2.txt};
			\addlegendentry{$p$}
			\addplot [color=red,solid,line width=2.0pt] table[x={t}, y={y}] {donnees_2.txt};
			\addlegendentry{$q$}
                         \addplot [color=black,densely dashdotted,line width=2.0pt] table[x={t}, y={z}] {donnees_2.txt};
                         \addlegendentry{$r$}
		\end{axis} 
	\end{tikzpicture}
	\caption{ \label{steady state} Convergence to the non-trivial steady state of Theorem \ref{thm:steady-state}. Simulation of an aggregate $(p,q,r)$ of initial size $(2,4,3)$ with parameters $\kappa_1= \kappa_2=\kappa_3=\kappa_{-1}=1$ and $\kappa_{-}=0.6$, implying $0<\bar\alpha<1$.}
\end{figure}

\begin{figure}[h!]
\centering
\begin{tikzpicture}

		\begin{axis}[xlabel={t}],
		
			\addplot [color=green,loosely dotted,line width=2.0pt] table[x={t}, y={x}] {donnees_3.txt};
			\addlegendentry{$p$}
			\addplot [color=red,solid,line width=2.0pt] table[x={t}, y={y}] {donnees_3.txt};
			\addlegendentry{$q$}
                         \addplot [color=black, densely dashdotted ,line width=2.0pt] table[x={t}, y={z}] {donnees_3.txt};
                         \addlegendentry{$r$}
		\end{axis} 
	\end{tikzpicture}
	\caption{\label{zero ss} Instability of the aggregate. Simulation of an aggregate $(p,q,r)$ of initial size $(2,4,3)$ with parameters $\kappa_1= \kappa_2=\kappa_3=\kappa_{-1} = 1$ and $\kappa_{-}=0.93$, implying $\bar\alpha>1$.}
\end{figure}
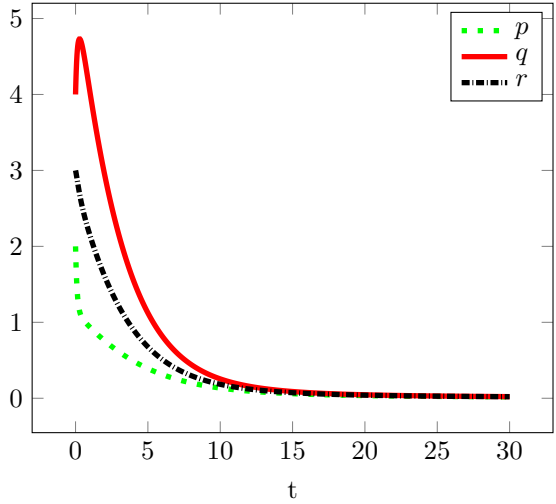

\begin{figure}[h!]
\centering
\begin{tikzpicture}

		\begin{axis}[xlabel={t}],
		
			\addplot [color=green,loosely dotted,line width=2.0pt] table[x={t}, y={x}] {donnees_1.txt};
			\addlegendentry{$p$}
			\addplot [color=red,solid,line width=2.0pt] table[x={t}, y={y}] {donnees_1.txt};
			\addlegendentry{$q$}
                         \addplot [color=black,densely dashdotted,line width=2.0pt] table[x={t}, y={z}] {donnees_1.txt};
                         \addlegendentry{$r$}
		\end{axis} 
	\end{tikzpicture}
	\caption{\label{growing polynomial} Growth of the aggregate. Simulation of an aggregate $(p,q,r)$ of initial size $(2,4,3)$ with parameters $\kappa_1= \kappa_2=\kappa_3=\kappa_{-1} = 1$ and $\kappa_{-}=0.2$, implying $\bar\alpha<0$.}
\end{figure}
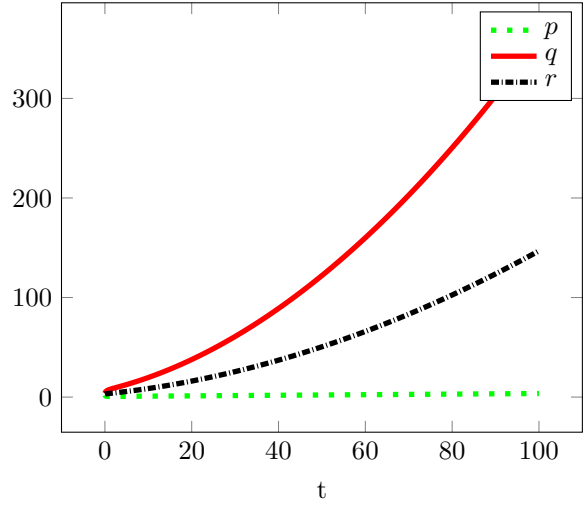

Closer inspection of the simulation results for growing aggregates (see Figure \ref{growing polynomial}) shows that
the growth is polynomial in time. This is verified by the following formal result.

\begin{thm}
With the notation of Theorem \ref{thm:steady-state}, if $\bar\alpha < 0$, then there exists a formal approximation of
a solution of \eqref{eq_1} of the form 
\be{ansatz}
   p(t)=p_1 t+ \textit{o}(t) \,,\qquad q(t)=q_2 t^{2}+ \textit{o}(t^{2}) \,,\qquad r(t) = r_2 t^{2}+ \textit{o}(t^{2}) \,,\qquad
   \mbox{as } t\to\infty \,,
\ee 
with
 \begin{equation}\label{def:poly}\begin{array}{l}
  p_1 = \frac{\kappa_- n}{\kappa_3 (2n\kappa_2 + \kappa_- n + 4\kappa_{-1})} \left( \kappa_1 \kappa_2 
     - \frac{\kappa_- n}{2(n-2)} \left( \kappa_1 + \kappa_{-1} + \frac{\kappa_- n(n-1)}{2(n-2)}\right)\right) >0\,, \\ 
  q_2 = \frac{n}{2}r_2 =  \frac{\kappa_3(n-2)(2n\kappa_2 + \kappa_- n + 4\kappa_{-1})}
      {\kappa_-(4\kappa_1 (n-2) + \kappa_- n^2)} p_1^2\,.
\end{array}\end{equation}
The approximation is (from a formal point of view) unique, including the choice of the exponents of $t$, among solutions
with polynomially or exponentially growing aggregate size $r$.
\end{thm}

\begin{proof}
Since $2r\le 2q \le nr$ holds for admissible states, when $r(t)$ tends to infinity, then also $q(t)$ tends to infinity at the same
rate, which we write with the {\em sharp order symbol} $O_s$ as
\be{q-equiv-r}
   q(t) = O_s(r(t)) \qquad\mbox{as } t\to\infty \,.
\ee
With $\alpha = \frac{s+p}{(n-2)r}$, we write the equations for $r$ and for $p+q$ as
\be{r_and_p+q}
   \dot r =  \kappa_2 p - (s+p) \frac{\kappa_- q}{(n-2)r} \,,\qquad    
   \dot p + \dot q = \kappa_1 s - p \left( \frac{\kappa_-(n-1)q}{(n-2)r} + \kappa_{-1}\right) \,.
\ee
Since the right hand sides have to be asymptotically nonnegative by the growth of $q$ and $r$, taking \eqref{q-equiv-r}
into account, the first equation implies $s(t) = O(p(t))$, and the second implies $p(t)=O(s(t))$, i.e.
\be{s-equiv-p}
   s(t) = O_s(p(t)) \qquad\mbox{as } t\to\infty \,.
\ee
If the growth were exponential, i.e. $r(t),q(t) = O_s(e^{\lambda t})$, $\lambda>0$, then \eqref{r_and_p+q} would imply
$p(t), s(t) = O_s(e^{\lambda t})$. Then the negative term $-\kappa_3 p(t)s(t) = O_s(e^{2\lambda t})$ in the first equation
in \eqref{eq_1}  could not be balanced by any of the positive terms, and would drive $p$ to negative values. This 
contradiction excludes exponential growth.

For polynomial growth, i.e. $r(t),q(t) = O_s(t^\gamma)$, \eqref{r_and_p+q} implies $p(t),s(t) = O_s(t^{\gamma-1})$.
In the equation for $q$ in \eqref{eq_1}, $\dot q$ and $p$ are small compared to $q$. Therefore it is necessary that
$s(t)p(t) = O_s(q(t))$, implying $2\gamma-2=\gamma$ and, thus, $\gamma=2$. This justifies the ansatz \eqref{ansatz}
with the addition $s(t) = s_1 t + o(t)$. Substitution into the differential equations and comparison of the leading-order terms
gives equations for the coefficients:
\begin{align*}
  \mbox{2nd equ. in \eqref{eq_1}:} \qquad & 0 = \kappa_3 p_1 s_1 - \kappa_- q_2 \,,\\
  \mbox{\eqref{eq:q-r}:} \qquad & 0 = \kappa_3 p_1 s_1 - \kappa_- q_2 \left(1 - \alpha(q_2,r_2)\right) \,,\\
  \mbox{1st equ. in \eqref{r_and_p+q}:} \qquad & 2r_2 =  \kappa_2 p_1 - (s_1+p_1) \frac{\kappa_- q_2}{(n-2)r_2} \,, \\
  \mbox{2nd equ. in \eqref{r_and_p+q}:} \qquad & 2q_2 = \kappa_1 s_1 - p_1 \left( \frac{\kappa_-(n-1)q_2}{(n-2)r_2} + \kappa_{-1}\right) \,,
\end{align*}
This system can be solved explicitly by first noting that the first two equations imply $\alpha(q_2,r_2)=0$ and, thus,
$2q_2=nr_2$. Using this in the third and fourth equation gives a linear relation between $p_1$ and $s_1$. This again can 
be used in the fourth equation to write $q_2$ as a linear function of $s_1$. The division of the first equation by $s_1$ then
gives the formula for $p_1$ in \eqref{def:poly}. The positivity of $p_1$ is a consequence of \eqref{def:f2}.
\end{proof}

For all the results so far the positivity of the rate constant $\kappa_-$ for breaking cross-links has been essential. 
Therefore it seems interesting to consider the special case $\kappa_- = 0$ separately. It turns out that the dynamics
is much simpler. The aggregate size always grows linearly with time.

\begin{thm} Let $3\le n\in\mathbb{N}$, $\kappa_1,\kappa_2,\kappa_3,\kappa_{-1}>0$, and $\kappa_{-} = 0$. 
Let $(p_0,q_0,r_0) \in (0,\infty)^3$ satisfy \eqref{cont-cond}. Then the solution of \eqref{eq_1} satisfies 
$$
  \lim_{t\to\infty}p(t) = p_{\infty}:= \frac{(n-2) \kappa_1 \kappa_2}{\kappa_3 ( \kappa_2 (n-2) + \kappa_{-1})} \,,\qquad
  \lim_{t\to\infty}s(t) = s_{\infty}:=  \frac{(n-2) \kappa_2}{2 \kappa_3} \,,
$$ $$
   q(t) = p_\infty (\kappa_2 + \kappa_3 s_\infty)t + o(t) \,,\qquad r(t) = \kappa_2 p_\infty t + o(t) \,,
     \qquad\mbox{as } t\to\infty \,.
$$
\end{thm}
\begin{proof}
For $\kappa_{-} = 0$ the right hand sides in \eqref{eq_1} depend only on $p$ and $s=nr-2q-p$, meaning that these two
variables solve a closed system: 
\begin{align*}
\dot p &= \kappa_1 s  - (\kappa_2 + \kappa_{-1} + \kappa_3 s)p \,,\\
\dot s &= ((n-1)\kappa_2 + \kappa_{-1})p  - (\kappa_1 + \kappa_3 p)s  \,.
\end{align*}
The unique nontrivial steady state $(p_{\infty},s_{\infty})$ can be computed explicitly.
We prove that it is globally attracting by constructing a Lyapunov functional. Let $a\ge 1$ and 
$$
   R_a := \left[ \frac{p_\infty}{a}, ap_\infty\right] \times \left[ \frac{s_\infty}{a}, as_\infty\right]\,.
$$ 
For each point $(p,s)\in (0,\infty)^2$ there is a unique value of $a\ge 1$ such that $(p,s)\in \partial R_a$.
Therefore the Lyapunov function
 \begin{align*}
     L(p,s) := a-1 \qquad\mbox{for } (p,s) \in \partial R_a \,,
 \end{align*}
 is well defined and definite in the sense $L(p,s)\ge 0$ with equality only for $(p,s)=(p_\infty,s_\infty)$. It remains to prove
 that the flow on $\partial R_a$ is strictly inwards. For example, for the left boundary part,
 $$
     \dot p \bigm|_{(p,s)\in \{p_\infty/a\}\times [s_\infty/a,as_\infty]} > 
     \left(\kappa_1 - \kappa_3 \frac{p_\infty}{a}\right)\frac{s_\infty}{a} - (\kappa_2 + \kappa_{-1})\frac{p_\infty}{a}
     = \frac{\kappa_3 p_\infty s_\infty (a-1)}{a^2} >0 \,,
 $$
 where for the first inequality it has been used that $p_\infty < \kappa_1/\kappa_3$, and the equality follows from the fact
 that $\dot p$ vanishes at the steady state. Similarly it can be shown that $\dot p<0$ on the right boundary part, $\dot s>0$
 on the lower  boundary part, and $\dot s < 0$ on the upper boundary part.
 
 The linear growth of $q$ and $r$ follows from
$$
\lim_{t\to\infty} \dot q(t) = \kappa_2 p_\infty + \kappa_3 p_\infty s_\infty  \,,\qquad
\lim_{t\to\infty}\dot r(t) =  \kappa_2 p_\infty \,.
$$
\end{proof}

This result shows that the breakage of cross-links has somewhat contradictory effects, depending on the parameter regime: 
It can speed-up the aggregation dynamics, producing a quadratic rather than linear growth of the aggregate size 
(Case 3 of Conjecture \ref{conj:cases}). This is linked to the fact that it allows the aggregates to rearrange in a more compact 
way. On the other hand, it may slow down the dynamics, such that the aggregate only reaches a finite size (Case 1) or even 
disintegrates completely (Case 2).

\section{Comparison with experimental data --- Discussion}

\paragraph{Comparison with experimental data:} There are only limited options for a serious comparison of the theoretical
results with experimental data. We shall use the data shown in Figure \ref{fig 2}, which have been published in
\cite{martens2}. It provides observed numbers of aggregates in dependence of ubiquitin for a fixed concentration of p62.
Our results do not permit a direct comparison with this curve, which would require modelling of the process of nucleation 
of aggregates. However, the data provide at least some information about concentration levels of ubiquitin and p62, such
that stable aggregates exist.

\begin{center}
\begin{figure}[h!]
\centering
    \includegraphics[width=6cm]{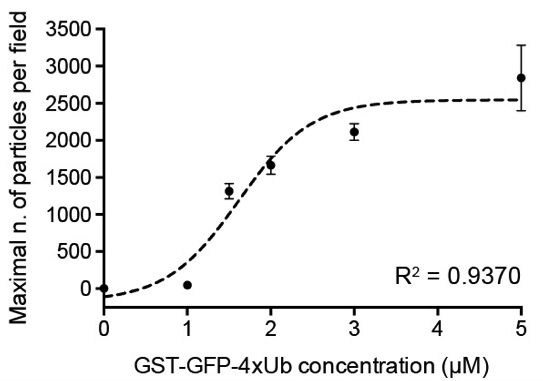}
  \caption{Number of  aggregates in terms of $[Ubi]$ (or more precisely $(4\times Ubi-GST-GFP)_2$) 
  at fixed [p62]$=2\mu M$  \cite{martens2}. 
  Average and SD among three independent replicates are shown. The dashed line represents a fitted
  sigmoidal (more precisely, logistic) function, centered around $[Ubi] = 1.6\mu M$. Note that here p62 monomers 
  are counted. Under the assumption that p62 only occurs in oligomers of size $n$ we have [p62]=$n$[p62$_n$].
  The regression coefficient R$^2$ measures the quality of the fit. }
  \label{fig 2}
\end{figure}
\end{center}

For meaningful quantitative comparisons with these scarce data we need to reduce the number of parameters in our model. 
As a first step, we fix
the value $n=5$ of the size of p62 oligomers, following \cite{martens2} where values between $5$ and $6$ for GFP-p62 have
been found (although we note that in \cite{10.7554/eLife.08941} an average of about $n=24$ has been reported for 
mCherry-p62 in vitro). This implies that the experiment corresponds to an oligomer concentration of 
$[{\rm p62}_5] = [p62]/5 = 0.4\mu M$. 

Concerning the rate constants, we make the assumption that the binding and, respectively, the unbinding rate constants are
equal, i.e. $\kappa_1' = \kappa_2' = \kappa_3'$ and $\kappa_{-1}=\kappa_-$. This will allow to express all our results in terms
of one {\em dissociation constant} $K_d := \kappa_{-1}/\kappa_1'$.

From Figure \ref{fig 2}  we conclude that for an oligomer concentration of $[{\rm p62}_5] = 0.4\mu M$ the growth of stable aggregates requires a cross-linker concentration $[Ubi]$ roughly between $0.6\mu M$ and $2.6\mu M$ ($(1.6\pm 1)\mu M$).
According to the results of the preceding section, these values should correspond to situations
with either $\bar\alpha=0$ or $\bar\alpha=1$, depending on the question, if the equilibrium aggregate sizes of
Case 1 in Conjecture \ref{conj:cases} are large enough to be detected in the experiment, or if we need to be in Case 3
of growing aggregates. Therefore, with the above assumptions, with $\kappa_1 = \kappa_1' [Ubi]$, 
$\kappa_2 = \kappa_2' [{\rm p62}_5]$, and with \eqref{def:f1}, \eqref{def:f2}, we obtain for $\bar\alpha=1$:
\be{eq:Kd}
   [{\rm p62}_n]\,[Ubi] = \frac{K_d}{n-2} \left( [Ubi] + \frac{(2n-3)K_d}{n-2}\right) \,,
\ee
and for $\bar\alpha = 0$:
\be{eq:Kd1}
   [{\rm p62}_n]\,[Ubi] = \frac{n K_d}{2(n-2)} \left( [Ubi] + \frac{(n^2 + n - 4)K_d}{2(n-2)}\right) \,.
\ee
Solving these equations for $K_d$ with $n=5$, $[{\rm p62}_n] = 0.4\mu M$, and with $[Ubi]$ between $0.6\mu M$ and 
$2.6\mu M$, gives estimates for $K_d$ between $0.44\mu M$ and $0.73\mu M$ for $\bar\alpha=1$, and between 
$0.20\mu M$ and $0.31\mu M$ for $\bar\alpha=0$. So we claim that at least the order of magnitude is significant. It differs by three orders of magnitude from published
data on the reaction between ubiquitin and the UBA domain of p62 ($K_d \approx 540\mu M$ \cite{PMID:19931284}).
This should not be so surprising, since in the context of growing aggregates the reactions can be strongly influenced by avidity effects.

\paragraph{Discussion:}
We return to Conjecture \ref{conj:cases}, where the long-time behaviour is described in terms of the value of the parameter
$\bar\alpha$ defined in \eqref{def:baralpha}. With the simplifying assumptions on the reaction rate constants from above,
the statements of the conjecture are depicted in Figure \ref{fig 1} for the fixed values $n=5$ and $K_d = 0.5\mu M$ 
(motivated by the estimates above) in a bifurcation diagram in terms of the concentrations $[Ubi]$ and $[{\rm p62}_n]$.
Note the unsymmetry in the dependence on the two quantities: The critical values for $[Ubi]$ tend to zero as 
$[{\rm p62}_n]$ tends to infinity, whereas the critical values for $[{\rm p62}_n]$ tend to the positive values 
$\frac{K_d}{n-2}$ for $\bar\alpha=1$ and $\frac{n K_d}{2(n-2)}$ for $\bar\alpha=0$, as $[Ubi]$ tends to infinity.

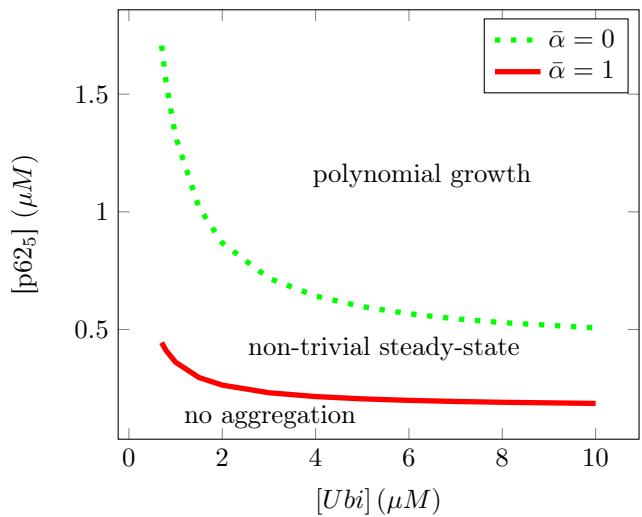
\begin{figure}[h!]
\centering
\begin{tikzpicture}
		\begin{axis}[
		 xlabel = {$[Ubi]\, (\mu M)$},
  ylabel = {[p62$_5$]  ($\mu M$)}]
			\addplot [color=green,loosely dotted,line width=2.0pt] table[x={x}, y={y}] {bif-diagram-data.txt};
			\addlegendentry{$\bar \alpha =0$}
			\addplot [color=red,solid,line width=2.0pt] table[x={x}, y={z}] {bif-diagram-data.txt};
			\addlegendentry{$\bar \alpha=1$}                      
\end{axis}
\draw (4,3.5) node {polynomial growth};
\draw (3.5,1.2) node {non-trivial steady-state};
\draw (2,0.3) node {no aggregation};
	\end{tikzpicture}

	\caption{Bifurcation diagram corresponding to Conjecture \ref{conj:cases}   for $n=5$, $K_d = 0.5 \mu M$.}
	\label{fig 1}
\end{figure}

There is a significant uncertainty concerning the oligomer size $n$, which has so far been assumed to be 5, according to
observations in \cite{martens2}. Actually, a distribution of oligomer sizes should be expected in the experiments of 
Figure \ref{fig 2} with the occurrence of much larger oligomers. For this reason the computation of $K_d$ from \eqref{eq:Kd} 
has been repeated for a range of values of $n$ between $n=3$ and $n=100$. The results are depicted in Figure 
\ref{frontier_1}, which shows that the predicted values of $K_d$ might be larger by up to an order of magnitude compared
to the case $n=5$, but still small
compared to \cite{PMID:19931284}, if larger oligomer sizes are considered and $\bar\alpha=1$ is relevant. 
The asymptotic behaviour for large oligomer sizes is easily seen to be $K_d = O(n^{1/2})$.
On the other hand,
if $\bar\alpha=0$ is relevant, the value of $K_d$ becomes smaller by up to an order of magnitude for large oligomers
with the asymptotic behaviour $K_d= O(n^{-1/2})$.

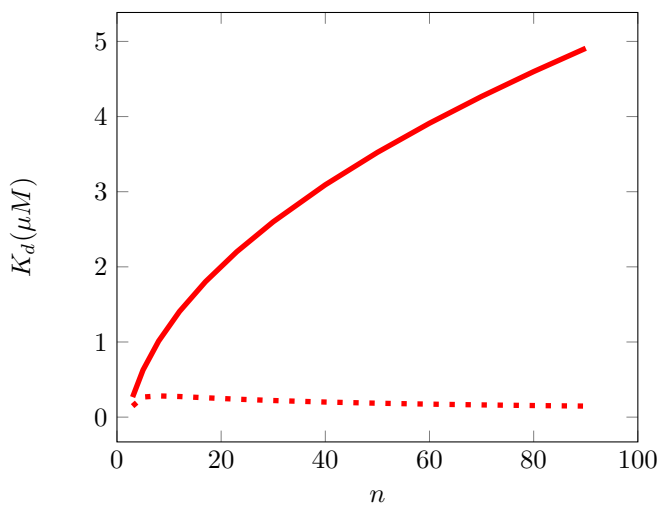
\begin{figure}[h!]
\centering
\begin{tikzpicture}
\begin{axis}[
  xlabel = {$n$},
  ylabel = {$K_d (\mu M)$},
  xmin=0, xmax=100,]
\addplot [color=red,loosely dotted,line width=2.0pt] table[x={x}, y={y}] {Kdn-data.txt};
\addplot [color=red,solid,line width=2.0pt] table[x={x}, y={z}] {Kdn-data.txt};
\end{axis}
\end{tikzpicture} 
	\caption{ The dissociation constant $K_d$ determined from \eqref{eq:Kd} (solid line) and \eqref{eq:Kd1} (dashed line), depending on the p62 oligomer size $n$. 
	Ubiquitin and p62 oligomer concentrations from Figure \ref{fig 2} at the onset of aggregation: $[Ubi]=1.6 \mu M$, 
	$[{\rm p62}_n]=0.4 \mu M$. \label{frontier_1} }
\end{figure}

\section{Conclusion}
In this article, we have proposed an ODE model for the growth and decay of aggregates of p62 oligomers cross-linked by 
ubiquitin chains. Under the assumption of unlimited supply of free oligomers and cross-linkers we found three possible asymptotic regimes: complete degradation of aggregates, convergence towards a finite aggregate size, and unlimited growth (quadratic in time) of the aggregate size. In the latter case, growing aggregates are asymptotically tightly packed with the 
maximum number of cross-links. These statements are supported by a mixture of explicit steady state computations, 
formal asymptotic analysis, and numerical simulations. The three regimes, which can be separated explicitly in terms
of the reaction constants, have been illustrated by the simulation results. Rigorous proofs of the long-time behaviour in the
three regimes are the subject of ongoing investigations.

A comparison of the theoretical results with data from \cite{martens2} has provided an estimate for the dissociation constant
of the elementary reaction between ubiquitin and the UBA domain of p62 in the context of growing aggregates.

There are several possible extensions of this work. A limitation of the original discrete model is that the description of
aggregates by triplets $(i,j,k)$ is very incomplete. Typically, very different configurations are described by the same triplet.
For example, we could imagine very homogeneous or very heterogeneous aggregates, \emph{i.e.} fully packed in certain regions and very loose in others. Reaction rates will strongly depend on the configuration, including information about the 
geometry of the aggregate. In principle one can imagine an attempt to overcome these difficulties based on a random graph
model \cite{frieze}, but the resulting model describing probability distributions on the sets of all possible aggregate shapes
would be prohibitively complex. An intermediate solution would be a more serious approach to finding formulas for quantities
like the probability $\alpha$ of losing an oligomer, when a cross-link breaks, based on typical probability distributions.

The model \eqref{eq_1} describes an intermediate stage of the aggregation process. On the one hand, the large aggregate
assumption means that we are dealing with the growth of already developed aggregates, neglecting the nucleation process,
which is important for the number of established aggregates. A model of the nucleation process would be based on the 
discrete representation and it would have to be stochastic. On the other hand, we neglect two effects important for a later
stage of the process. The first and obvious one is the limited availability of free p62 oligomers and ubiquitin cross-linkers.
It would be rather straightforward to incorporate this into the model, however at the expense of increased complexity.
It would also eliminate the dichotomy between the Cases 1 and 3 of Conjecture \ref{conj:cases} since unbounded growth
would be impossible. For relatively large initial concentrations of free particles, one could imagine a two-time-scale behaviour
with an initial quadratic growth and saturation on a longer time scale.
The other effect, which is neglected here but definitely present in experiments, is coagulation of aggregates. This is the subject
of ongoing work, based on the PDE model \eqref{PDE} derived in the appendix and enriched by an account of the
coagulation process.

\begin{appendix}
\section{Large aggregate limit}
\label{app:model}
We denote by $c_{i,j,k}(t)$ the probability of the aggregate to be in the state $(i,j,k)$ at time $t$. Its evolution will be 
determined by a jump process model of the reactions with the rates given in \eqref{def:r1}, \eqref{def:r2}, \eqref{def:r3}, 
\eqref{def:r-1}, \eqref{def:alpha-discr}, \eqref{def:r-2}, \eqref{def:ell}, and \eqref{def:r-3}. 

For this purpose the relation between pre-reaction state $(i',j',k')$ and post-reaction state $(i,j,k)$ needs to be inverted. 
This is easy except for Reaction 5, where we have $j=j'-1$, $k=k'-1$, and, with \eqref{def:ell},
\be{pre2post}
   i = i' + 1 - \ell_{i',j',k'} = i' + 1 - \left\lfloor \frac{(n-1)i'}{nk'-2j'} \right\rceil\,.
\ee
The inversion is not possible in general. Occasionally, $\ell_{i',j',k'}$ will increase by one, when $i'$ is increased by one,
implying that $i$ might take the same value for two consecutive values of $i'$. Even worse:
For the extreme case $nk'-2j'=n-1$, where after the loss of a p62 oligomer all binding sites 
are busy with two-hand bound $Ubi$ except the one remaining after breaking the connection, i.e. $nk-2j=1=i$. This
state is independent from the number $i'\in\{0,\ldots,n-1\}$ of one-hand bound $Ubi$ getting lost with the oligomer.
Therefore we introduce
$$
  I_{i,j,k} = \left\{ i':\, i = i' + 1 - \ell_{i',j+1,k+1}\right\}
$$
The equation for the probability distribution reads
\begin{eqnarray}
\frac{dc_{i,j,k}}{dt} &=& (r_1c)_{i-1,j,k} - (r_1c)_{i,j,k} + (r_2c)_{i+1,j-1,k-1} - (r_2c)_{i,j,k} + (r_3c)_{i+1,j-1,k} - (r_3c)_{i,j,k} 
   \nonumber\\
 && + (r_{-1}c)_{i+1,j,k} - (r_{-1}c)_{i,j,k} +  \sum_{i'\in I_{i,j,k}} (r_{-2}c)_{i',j+1,k+1} - (r_{-2}c)_{i,j,k} \nonumber\\
  && + (r_{-3}c)_{i-1,j+1,k} - (r_{-3}c)_{i,j,k} \,. \label{eq}
\end{eqnarray} 
We introduce a typical value $k_0$ for the number $k$ of oligomers in the aggregate and use it also as a reference value
for $i$ and $j$, leading by the definition  \eqref{def:pqr} to the scaled triplet $(p,q,r)$. The latter lives on a grid with spacing
$\Delta p = \Delta q = \Delta r := 1/k_0$ and, thus, becomes a continuous variable in the large aggregate limit $k_0\to\infty$. Therefore we postulate
the existence of a probability density $P(p,q,r,t)$ such that
$$
   c_{i,j,k}(t) \approx k_0^3\, P\left(\frac{i}{k_0}, \frac{j}{k_0}, \frac{k}{k_0}, t\right) \,.
$$
Division of \eqref{eq} by $k_0^3$ and the limit $k_0\to\infty$ ($\Delta p = \Delta q = \Delta r \to 0$) will lead to an equation 
for $P$. We deal with the six differences on the right hand side of \eqref{eq}, corresponding to the six reactions, separately. \\
{\em Reaction~1:}
\begin{align*}
   & k_0^{-3} \left[ (r_1c)_{i-1,j,k} - (r_1c)_{i,j,k}\right] \\
   & \approx \frac{1}{\Delta p} \left[ \kappa_1(nr-p+\Delta p - 2q)P(p-\Delta p,q,r,t) - \kappa_1 (nr-p-2q)P(p,q,r,t)\right] \\
   & \to -\partial_p (\kappa_1(nr-p-2q)P) \,.
\end{align*}
{\em Reaction 2:}
\begin{align*}
   & k_0^{-3} \left[ (r_2c)_{i+1,j-1,k-1} - (r_2c)_{i,j,k} \right] \\
   & \approx \frac{1}{\Delta p} \left[ \kappa_2(p+\Delta p)P(p+\Delta p,q-\Delta q,r-\Delta r,t) 
          - \kappa_2 p P(p,q-\Delta q,r-\Delta r,t)\right] \\
   & \quad + \frac{1}{\Delta q} \left[ \kappa_2p P(p,q-\Delta q,r-\Delta r,t) - \kappa_2 p P(p,q,r-\Delta r,t)\right] \\
   & \quad + \frac{1}{\Delta r} \left[ \kappa_2p P(p,q,r-\Delta r,t) - \kappa_2 p P(p,q,r,t)\right] \\
   & \to \partial_p (\kappa_2 pP) - \partial_q (\kappa_2 pP) - \partial_r (\kappa_2 pP)\,.
\end{align*}
{\em Reaction 3:} Since this is a second-order reaction, it would dominate the dynamics for large $k_0$, if the reaction 
constant were of the same order of magnitude as the others. In order to avoid this, we set 
$\kappa_3' = \kappa_3/k_0$ and keep $\kappa_3$ fixed as $k_0\to\infty$. 
\begin{align*}
   & k_0^{-3} \left[ (r_3c)_{i+1,j-1,k} - (r_3c)_{i,j,k} \right] \\
   & \approx \frac{1}{\Delta p} \bigl[ \kappa_3(p+\Delta p)(nr-p-\Delta p-2q+2\Delta q)P(p+\Delta p,q-\Delta q,r,t) \\
   &\qquad\quad    - \kappa_3 p(nr-p-2q+2\Delta q) P(p,q-\Delta q,r,t)\bigr] \\
   & \quad + \frac{1}{\Delta q} \left[ \kappa_3 p(nr-p-2q+2\Delta q) P(p,q-\Delta q,r,t) - \kappa_3 p(nr-p-2q) P(p,q,r,t)\right] \\
   & \to \partial_p (\kappa_3 p(nr-p-2q)P) - \partial_q (\kappa_3 p(nr-p-2q)P) \,.
\end{align*}
{\em Reaction 4:}
\begin{align*}
   & k_0^{-3} \left[ (r_{-1}c)_{i+1,j,k} - (r_{-1}c)_{i,j,k} \right] \\
   & \approx \frac{1}{\Delta p} \left[ \kappa_{-1}(p+\Delta p)P(p+\Delta p,q,r,t) 
          - \kappa_{-1} p P(p,q,r,t)\right] \\
   & \to \partial_p (\kappa_{-1} pP) \,.
\end{align*}
{\em Reaction 5:} As preparatory steps, we compute
$$
  \ell_{i',j+1,k+1} = \left\lfloor \frac{(n-1)i'}{nk-2j+n-2}\right\rceil = \left\lfloor \frac{(n-1)p'}{nr-2q+(n-2)\Delta p}\right\rceil\,.
$$
As a function of $p'$, this is piecewise constant and equal to $\left\lfloor \frac{(n-1)p'}{nr-2q}\right\rceil$ with jumps 
(for small $\Delta p$) close to the set $\frac{nr-2q}{n-1}\left(\frac{1}{2} + \mathbb{N}_0\right)$. Away from these points
the map \eqref{pre2post} from pre- to post-reaction states is invertible with 
$$
   p' =  p + \Delta p \left(\ell(p,q,r) - 1\right) \,,\qquad \ell(p,q,r) := \left\lfloor \frac{(n-1)p}{nr-2q}\right\rceil \,.
$$
Note that $p'$ has been replaced by $p$ in the argument of $\ell$ since $p-p'=O(\Delta p)$. At all these generic points the
sum in \eqref{eq} has only one term. We shall also need
$$
  \alpha_{j,k} \to \frac{nr-2q}{(n-2)r} =: \alpha(q,r) \,.
$$
Thus,
\begin{align*}
   & k_0^{-3} \left[ \mathbbm{1}_{i\ge 1}(r_{-2}c)_{I_{i,j,k},j+1,k+1} - (r_{-2}c)_{i,j,k} \right] \\
   & \approx \frac{1}{\Delta p} \bigl[ \kappa_- \alpha(q+\Delta q,r+\Delta r)(q+\Delta q)
     P(p+\Delta p(\ell-1),q+\Delta q,r+\Delta ,t) \\
    & \qquad\quad - \kappa_- \alpha(q+\Delta q,r+\Delta r)(q+\Delta q) P(p,q+\Delta q,r+\Delta r,t)\bigr] \\
    & \approx \frac{1}{\Delta q} \bigl[ \kappa_- \alpha(q+\Delta q,r+\Delta r)(q+\Delta q) P(p,q+\Delta q,r+\Delta ,t) \\
    & \qquad\quad - \kappa_- \alpha(q,r+\Delta r)q P(p,q,r+\Delta r,t)\bigr] \\  
    & \approx \frac{1}{\Delta r} \bigl[ \kappa_- \alpha(q,r+\Delta r)q P(p,q,r+\Delta ,t) \\
    & \qquad\quad - \kappa_- \alpha(q,r)q P(p,q,r,t)\bigr] \\  
   & \to \partial_p (\kappa_- \alpha (\ell-1)q P) + \partial_q (\kappa_- \alpha q P) + \partial_r (\kappa_- \alpha q P)\,.
\end{align*}
Note that the factor $\ell-1$ has been written inside the derivative since $\ell$ is constant away from finitely many
critical points. We replace a detailed analysis at these points by the simple argument that the equation for $P$ has to be 
in conservation form to preserve the total probability. Finally we introduce a simplification by dropping the rounding 
operation in $\ell$.\\
{\em Reaction 6:}
\begin{align*}
   & k_0^{-3} \left[ (r_{-3}c)_{i+1,j-1,k} - (r_{-3}c)_{i,j,k} \right] \\
   & \approx \frac{1}{\Delta p} \bigl[ \kappa_- (1-\alpha(q+\Delta q,r))(q+\Delta q)P(p-\Delta p,q+\Delta q,r,t) \\
   &\qquad\quad    - \kappa_- (1-\alpha(q+\Delta q,r))(q+\Delta q) P(p,q+\Delta q,r,t)\bigr] \\
   & \quad + \frac{1}{\Delta q} \left[ \kappa_- (1-\alpha(q+\Delta q,r))(q+\Delta q) P(p,q+\Delta q,r,t)
      - \kappa_- (1-\alpha(q,r))q P(p,q,r,t)\right] \\
   & \to -\partial_p (\kappa_- (1-\alpha)qP) + \partial_q (\kappa_- (1-\alpha)qP) \,.
\end{align*}
Collecting our results, the limiting equation for the evolution of $P$ reads
\begin{eqnarray}
  &&\partial_t P + \partial_p \left( \left((\kappa_1 - \kappa_3 p)(nr-p-2q) - (\kappa_2 + \kappa_{-1})p 
  + \kappa_- q \left(1 - \frac{(n-1)p}{(n-2)r}\right) \right)P\right) \nonumber\\
  &&+ \partial_q \left( (\kappa_2 p+\kappa_3 p(nr-p-2q) - \kappa_- q)P\right)
   + \partial_r \left( (\kappa_2 p - \kappa_- \alpha q)P\right) = 0 \,.\label{PDE}
\end{eqnarray}
For deterministic initial conditions of the form $P(p,q,r,0) = \delta(p-p_0)\delta(q-q_0)\delta(r-r_0)$ the state remains
deterministic: $P(p,q,r,t) = \delta(p-p(t))\delta(q-q(t))\delta(r-r(t))$, where $(p(t),q(t),r(t))$ solves the initial value problem
\eqref{eq_1}.

\end{appendix}

\bibliography{mybib}{}
\bibliographystyle{plain}

\end{document}